\documentclass[prd,nofootinbib,preprintnumbers,preprint]{revtex4}
\usepackage{amsmath}
\usepackage{amsfonts}
\usepackage{amssymb}
\usepackage{dsfont}
\usepackage[hyperfootnotes=false]{hyperref}
\usepackage[dvipsnames]{xcolor}

\newcommand{\mX}{\mathcal X}
\newcommand{\mY}{\mathcal Y}
\newcommand{\mZ}{\mathcal Z}
\newcommand{\mK}{\mathcal K}

\newcommand{\mF}{\mathcal F}
\newcommand{\mG}{\mathcal G}

\newcommand{\RP}{\right\rangle}
\newcommand{\LP}{\left\langle}
\newcommand{\beq}{\begin{equation}}
\newcommand{\eeq}{\end{equation}}
\newcommand{\bAL}{\begin{align}}
\newcommand{\eAL}{\end{align}}

\newcommand{\LCM}{\nabla}

\newcommand{\RIM}{Ric}

\newcommand{\TG}{\dot{\gamma}}

\newcommand{\LCS}{\mathcal D}

\newcommand{\tr}{\text{tr}}

\newcounter{example}[section]
\newcounter{remark}[section]
\newcounter{theorem}[section]
\newcounter{proposition}[section]
\newcounter{lemma}[section]
\newcounter{corollary}[section]
\newcounter{definition}[section]

\def\theremark{\arabic{section}.\arabic{remark}}
\def\thetheorem{\arabic{section}.\arabic{theorem}}

\def\thedefinition{\arabic{section}.\arabic{definition}}

\renewcommand*{\email}[1]{\footnote{Electronic address: \href{mailto:#1}{\nolinkurl{#1}} }}

\newenvironment{proof}{\noindent {\textit{Proof:}}
}{$\Box$\medskip}

\newenvironment{theorem}{\refstepcounter{theorem}
\medskip\noindent{\bf Theorem \thetheorem}:}{$\Box$\medskip}
\newenvironment{proposition}{\refstepcounter{theorem}\medskip\noindent{\bf
Proposition \thetheorem}:}{$\Box$\medskip}

\newenvironment{definition}{\refstepcounter{definition}\medskip\noindent{\bf
Definition \thedefinition}:}{$\Box$\medskip}

\begin{document}

\title{ 
Killing tensors and photon surfaces in foliated spacetimes}
\author{Kirill Kobialko\email{kobyalkokv@yandex.ru}}
\author{Igor Bogush\email{igbogush@gmail.com}}
\author{Dmitri Gal'tsov\email{galtsov@phys.msu.ru}}
\affiliation{Faculty of Physics, Moscow State University, 119899, Moscow, Russia}

\begin{abstract}
We present a purely geometric method for constructing a rank two Killing tensor in a spacetime with a codimension one foliation that lifts the trivial Killing tensors from slices to the entire manifold. The resulting Killing tensor can be nontrivial. 
A deep connection is found between the existence of such a Killing tensor and the presence of generalized photon surfaces in spacetime with two Killing vector fields. This technique generates Killing tensors in a purely algebraic way, without solving differential equations.  The use of our method is demonstrated for  Kerr, Kerr-Newman-NUT-AdS metrics and  Kerr-NUT-AdS multicharge gauged supergravity solution. 
 
\end{abstract}


\maketitle

\setcounter{page}{2}

\setcounter{equation}{0}
\setcounter{subsection}{0}
\setcounter{section}{0}

\section{Introduction}
Killing vectors explicitly describe the symmetries of a spacetime manifold equipped with a metric, which are related to its isometries. They provide a set of integrals of motion for geodesics and wave operators in field theories. A natural generalization of Killing vectors is the Killing tensors associated with other hidden symmetries of spacetime  \cite{Frolov:2017kze} which give conserved quantities of higher order in conjugate momenta. The first striking example of such a conserved quantity, quadratic in the particle momentum, is Carter's constant \cite{Carter:1968ks}, found for the Kerr metric. Although the rotating Kerr solution does not have spherical symmetry, the Carter constant is a generalization of the square of the angular momentum. However, the process of constructing  non-trivial Killing tensors is much more complicated than that for Killing vectors. Most of the known results  stem from the fact that the Killing tensor equation simplifies in certain special spacetimes, such as a space with a warped/twisted product structure \cite{Krtous:2015ona}, a space admitting a hypersurface orthogonal Killing vector field \cite{Garfinkle:2010er,Garfinkle:2013cha}, or special conformal Killing fields \cite{Koutras,Barnes:2002pm}.

In this article, we present a new purely geometric method for generating Killing tensors in manifolds with foliation of codimension one, based on lifting the Killing tensors constructed  in slices with an arbitrary second fundamental from. This approach does not require slices to be orthogonal to the Killing vector field or a spacetime to have the warped/twisted product structure. We present general equations for lifting of an arbitrary Killing tensor, which reduce to the result of Ref. \cite{Garfinkle:2010er} in the case of totally geodesic slices.  

We demonstrate mechanism of the emergence of a nontrivial Killing tensor in a total manifold from the trivial ones in slices constructed with the help of the Killing vectors tangent to them. We completely integrate the arising equations, find general compatibility and integrability conditions, and prove the theorem \ref{KBG}  providing  a generating technique for non-trivial Killing tensors, which is a purely geometric analogue of the method of separation of variables in the Hamilton-Jacobi formalism. 

Furthermore, we investigate relationship between the resulting slice structure and the structure of fundamental photon surfaces introduced in \cite{Kobialko:2020vqf}, which are a natural generalization of photon spheres \cite{Yoshino1,Yoshino:2019dty,Yoshino:2019mqw,Claudel:2000yi,Teo,Galtsov:2019fzq,Galtsov:2019bty,Cao:2019vlu}. The concept of photon spheres and surfaces plays a crucial role in study of the black hole shadows \cite{Grover:2017mhm,Grenzebach,Grenzebach:2015oea,Cunha:2018acu,Shipley:2019kfq}, Penrose inequalities \cite{Shiromizu:2017ego,Feng:2019zzn,Yang:2019zcn}, uniqueness \cite{Cederbaum,Yazadjiev:2015hda,Yazadjiev:2015jza,Yazadjiev:2015mta,Rogatko,Cederbaumo,Yoshino:2016kgi,Koga:2020gqd} and integrability  \cite{Pappas:2018opz,Glampedakis:2018blj} theorems. The integrability conditions for Killing tensors ensure the foliation slices to be fundamental photon surfaces if the generalized photon region inequalities are fulfilled. This constitutes one of the most important results we obtain here which  generalizes that of Ref. \cite{Koga:2020akc} for the case of stationary spaces without resorting to null geodesic equations. We hope that the result obtained helps to reveal  fundamental connection between Killing tensors and the structure of photon surfaces.

We apply this new technique to Kerr, Kerr-Newman-NUT-AdS and Kerr-NUT-AdS multicharge gauged supergravity \cite{Chong:2004na} solutions showing that this technique allows to obtain Killing tensors \cite{Kubiznak:2007kh,Vasudevan:2005bz}  purely algebraically without solving any differential equations. For these spacetimes we reveal the nature of Killing tensors hidden symmetry as arising from isometries in low-dimensional slices of a smooth foliation.

In Sec. \ref{sec:killing} we briefly describe the equations for the Killing vectors and Killing tensors of rank two. In Sec. \ref{sec:foliation} we consider spacetimes with foliation of codimension one and derive the equations governing the interplay between symmetries in the bulk and in the slices. In Sec. \ref{sec:generation} we describe the generating technique for such spacetimes. In Sec. \ref{sec:photon} we reveal the connection between the Killing tensors and the fundamental photon surfaces. Sec. \ref{sec:axially} provides some examples with axial symmetry. The appendices contain the proofs of the statements made in the main text and some generalizations.

\section{Setup}
\label{sec:killing}

Let $M$ be a Lorentzian manifold of dimension $m$ with scalar product $\LP \;\cdot\; ,\;\cdot\;\RP$ and Levi-Civita connection $\LCM$\footnote{Here we also use the notation $\underset{\mX \leftrightarrow \mY}{{\rm Sym}}\left\{ B(\mX, \mY) \right\}= B(\mX, \mY)+ B(\mY, \mX)$ and $[\mX, \mY]= \LCM_\mX\mY -\LCM_\mY\mX$.}.

\begin{definition} 
A vector field $\mK:M\rightarrow TM$ is called a Killing vector field if \cite{Chen} 
\begin{align}\label{eq:killing_equation}
\underset{\mX \leftrightarrow \mY}{{\rm Sym}}\left\{\LP \LCM_\mX\mK, \mY \RP\right\}=0, \quad \forall \mX,\mY\in TM. 
\end{align}
\end{definition} 

\begin{definition} 
A linear self-adjoint mapping $K(\;\cdot\;):TM\rightarrow TM$ is called a Killing mapping if
\begin{align}\label{eq:killing_mapping_equation}
\underset{\mX \leftrightarrow \mY\leftrightarrow \mZ}{{\rm Sym}}\left\{\LP\LCM_\mX K(\mZ),\mY\RP\right\}=0, \quad \forall \mX,\mY,\mZ\in TM,
\end{align} 
where the linear mapping $\LCM_\mX K(\;\cdot\;):TM\rightarrow TM$ is defined as follows
\begin{align}
\LCM_\mX K( \mY)\equiv\LCM_\mX( K( \mY))-K(\LCM_\mX (\mY)), \quad \forall \mX,\mY\in TM.
\end{align} 
\end{definition} 

One can introduce a Killing tensor as a symmetric form $K(\mX,\mY)=\LP K(\mX),\mY\RP$, which is associated with the conservation law quadratic in momenta. Indeed, consider the functions $\mathcal Q_{\mathcal K}(\;\cdot\;):TM \rightarrow\mathbb R$ and $\mathcal Q_{K}(\;\cdot\;,\;\cdot\;):TM\oplus TM\rightarrow\mathbb R$ defined as
\begin{align}
\mathcal Q_{\mathcal K}(\mX)\equiv\LP\mathcal K, \mX\RP, \quad \mathcal Q_{K}(\mX,\mY)\equiv\LP K(\mX),\mY\RP, \quad \forall \mX,\mY\in TM.
\end{align} 

\begin{proposition}
Let $\gamma:\mathbb{R}\rightarrow M$ be an affinely parameterized geodesic, i.e $\LCM_{\TG} \TG=0$. Then, the functions $\mathcal Q_{K}$ and $\mathcal Q_{\mathcal K}$ induce the conserved quantities along $\gamma$:
\begin{align}
\frac{d}{ds}\mathcal Q_{\mathcal K}(\TG)=0, \quad \frac{d}{ds}\mathcal Q_{K}(\TG,\TG)=0.
\end{align}
\end{proposition}

\begin{proof}
\begin{subequations}
\begin{align}
\frac{d}{ds}\mathcal Q_{\mathcal K}(\TG)\equiv\TG\LP\mathcal K,\TG \RP=\frac{1}{2}\cdot\underset{\TG \leftrightarrow \TG}{{\rm Sym}}\left\{\LP\LCM_{\TG}\mathcal K,\TG \RP\right\}+\LP\mathcal K,\LCM_{\TG}\TG \RP=0,
\end{align}
\begin{align}
\frac{d}{ds}\mathcal Q_{K}(\TG,\TG)\equiv\TG\LP K(\TG),\TG\RP=\frac{1}{6}\cdot\underset{\TG \leftrightarrow \TG\leftrightarrow \TG}{{\rm Sym}}\LP \LCM_{\TG}K(\TG),\TG\RP+\LP K(\LCM_{\TG}\TG),\TG\RP+\LP K(\TG),\LCM_{\TG}\TG\RP=0.
\end{align} 
\end{subequations}
\end{proof}

\begin{proposition}\label{prop:trivial_tensor}
Let $\mathcal K_\alpha$ be a set of $n$ Killing vector fields. Then, one can define the following trivial Killing mapping                  
\begin{align}\label{eq:trivial_tensor}
K(\mX)=\alpha \mX+\sum^n_{\alpha,\beta=1}\gamma^{\alpha\beta}\LP \mX,\mathcal K_\alpha\RP\mathcal K_\beta, \quad \gamma^{\alpha\beta}=\gamma^{\beta\alpha},
\end{align} 
where $\alpha$ and $\gamma^{\alpha\beta}$ is the set of $n(n+1)/2+1$ independent constants in $M$.
\end{proposition}

\begin{proof}
The Killing mapping is linear with respect to its argument $\mX$, and the Killing equations are linear with respect to the mapping $K$. Thus, let us prove Proposition \ref{prop:trivial_tensor} for both terms separately. For the first term, the mapping $K$ is proportional to the identity $K(\mX)=\alpha \mX$.
Then, the Killing equations give
\begin{align}
\underset{\mX \leftrightarrow \mY\leftrightarrow \mZ}{{\rm Sym}}\left\{\LP \mX \left(\alpha\right)\mY,\mZ\RP\right\}=0 \quad \Rightarrow \quad \mX(\alpha)=0.
\end{align} 
For the second term, the mapping $K$ defines a projection onto the subspace spanned by Killing vector fields $\mathcal K_{\alpha}$, i.e. $K(\mX)=\gamma^{\alpha\beta}\LP \mX,\mathcal K_{\alpha}\RP \mathcal K_{\beta}$
\begin{align}
\underset{\mX\leftrightarrow \mY\leftrightarrow \mZ}{{\rm Sym}}\left\{2\gamma^{\alpha\beta}\LP \mY,\LCM_\mX\mathcal K_{\alpha}\RP\LP  \mathcal K_{\beta},\mZ\RP\right\}+\underset{\mX\leftrightarrow \mY\leftrightarrow \mZ}{{\rm Sym}}\left\{\mX\left(\gamma^{\alpha\beta}\right)\LP \mY,\mathcal K_{\alpha}\RP\LP  \mathcal K_{\beta},\mZ\RP\right\}=0 \quad \Rightarrow \quad \mX(\gamma^{\alpha\beta})=0,
\end{align} 
where we used Killing equations again.
\end{proof}

Note that the trivial Killing mapping does not give new conservation laws. Indeed, in this case
\begin{align}
\mathcal Q_{K}(\TG,\TG)=\alpha \LP\TG,\TG\RP+\sum^n_{\alpha,\beta=1}\gamma^{\alpha\beta}\mathcal Q_{\mathcal K_\alpha}(\TG)\mathcal Q_{\mathcal K_\beta}(\TG),
\end{align} 
and hence the conserved quantity $\mathcal Q_{K}(\TG,\TG)$ is nothing else than a combination of the Killing vectors and a constant coefficient. However, one can show the existence of manifolds with nontrivial Killing tensors, which are not associated with the manifold isometries directly.

\section{Foliation lift}
\label{sec:foliation}

\begin{definition}
A {\em hypersurface} $S$ is an image of isometric embedding $i:S\rightarrow M$ of $m-1$ and $m$ dimensional pseudo Riemannian manifolds $S$ and $M$ respectively.
\end{definition}

A timelike/spacelike hypersurface have a unique unit normal vector field $\xi$ ($\LP \xi, \xi \RP\equiv\epsilon=\pm1$), while the second fundamental form $\sigma( \;\cdot\; ,\;\cdot\;):TS\times TS\rightarrow \mathbb R$ and the mean curvature are defined as follows
\begin{align}
\sigma(X,Y)\equiv \epsilon\LP \LCM_XY,\xi\RP, \quad X,Y\in TS, \quad  h={\rm Tr}(\sigma)/(m-1).
\label{SFF1}
\end{align} 
In this notation Gauss decomposition reads
\begin{align}
&\LCM_{X}Y=\LCS_XY+\sigma(X,Y)\cdot\xi,
\end{align} 
where $\LCS$ is the Levi-Civita connection in $S$ (see Ref. \cite{Chen} for details). 

\begin{definition} 
A {\em foliation} $\mF_\Omega$ of the manifold $M$ (with codimension one) parameterized by $\Omega\in\mathbb{R}$ with the lapse function $\varphi$ is a smooth family of hypersurfaces $S_\Omega$ (slices), whose union is $M$, which  satisfies the flow equation
\begin{align}
\frac {d \mathcal F_\Omega} {d \Omega}=\varphi \xi \quad \text{and} \quad  \frac {d\xi} {d \Omega}= - \epsilon\cdot \LCS \varphi,
\label{fl1}
\end{align}
where $\xi$ is a unit spacelike/timelike vector field normal to each slice $S_\Omega$.
\end{definition} 

\begin{proposition}\label{prop:projected_killing_equations}
Let $\mK$ be a Killing vector field on the manifold $M$, foliated by $\mF_\Omega$, with the normal $k_N\xi$ and tangent $\mK_\Omega\in TS_\Omega$ components, viz. $\mK=\mK_\Omega+k_N\xi$. Then, from the Killing equations (\ref{eq:killing_equation}) and the definitions (\ref{SFF1}),  (\ref{fl1}) it follows:
\begin{subequations}
\begin{align}
    &
      \underset{X \leftrightarrow Y}{{\rm Sym}}\left\{\LP \LCS_X \mathcal K_\Omega, Y\RP\right\}
    = 2\epsilon k_N\cdot{}^\Omega\sigma (X,Y), \label{eq:kkv19a}
    \\ &
      \LP\LCM_\xi\mathcal K_{\Omega}, X\RP
    = \epsilon\cdot\left(
          k_N \cdot \LCM_X\ln \varphi- X (k_N)
        - {}^\Omega\sigma (X,\mathcal K_{\Omega})
         \right),\label{eq:kkv19b}
    \\ &
      \LP \LCM_\xi \mathcal K_{\Omega}, \xi\RP
    = \epsilon\cdot\LCM_{\mathcal K_{\Omega}}\ln \varphi
    = -\epsilon\cdot\xi (k_N), \label{eq:kkv19c}
\end{align}
\end{subequations}
where $X,Y\in TS_{\Omega}$.
\end{proposition}

As a consequence of the Proposition \ref{prop:projected_killing_equations}, the Killing vector field in $M$, tangent to all the foliation slices $S_\Omega$, is also a Killing vector field in $S_\Omega$. In the general case, the projection $\mathcal K_{\Omega}$ is not a Killing vector in the slices of foliation since the right side of the equation (\ref{eq:kkv19a}) does not vanish. An exception is the case of totally umbilical or totally geodesic slice, where the projection of any Killing field is a conformal or ordinary Killing vector field respectively. Such slices arise if the field generating the foliation is a (conformal) Killing field and/or the spacetime has the structure of a warped/twisted product \cite{Chen}. As we will see further, the case of Killing tensors is more intricate.

Assuming that the tangent component $\mathcal{K}_\Omega$ forms the Killing vector in all slices $S_\Omega$, the left side of the equation (\ref{eq:kkv19a}) is equal to zero, so either ${}^\Omega\sigma(X,Y)$ or $k_N$ must also be zero. In the first case ${}^\Omega\sigma(X,Y)=0$ slices are totally geodesic hypersurfaces, and the equations (\ref{eq:kkv19b}), (\ref{eq:kkv19c}) reduce to
\begin{equation}
      \LP \LCM_\xi \mathcal K_{\Omega}, X\RP
    = \epsilon\cdot\left(
          k_N \cdot \LCM_X\ln \varphi- X (k_N)
          \right),\qquad
      \LCM_{\mathcal K_{\Omega}}\ln \varphi = -\xi (k_N).
\end{equation}
In the second case of a trivial normal component $k_N=0$, these equations are
\begin{equation}\label{eq:kkv2}
      \LP \LCM_\xi \mathcal K_{\Omega} , X\RP
    = -\epsilon\cdot{}^\Omega\sigma (X,\mathcal K_{\Omega}),
    \qquad
    \LCM_{\mathcal K_{\Omega}} \varphi =0.
\end{equation}
Thus, the generation of the Killing vectors in $M$ from $\mathcal{K}_\Omega$ with a nontrivial normal component $k_N$ is possible in the case of the totally geodesic slices only.

\begin{proposition} \label{prop:killing_map}
Let $K$ be a Killing mapping on $M$, such that $K(\; \cdot\;)=K^{(\; \cdot\;)}_\Omega+k^{(\; \cdot\;)}_N\xi$, where $k^{(\; \cdot\;)}_N\xi$ is a normal component and $K^{(\; \cdot\;)}_\Omega\in TS_\Omega$ is a tangent component. Then, the Killing equations (\ref{eq:killing_mapping_equation}) split into the following parts
\begin{subequations}\label{eq:killing_map_xy}
\begin{align}
    & 
      \underset{X \leftrightarrow Y\leftrightarrow Z}{{\rm Sym}}\left\{\LP\LCS_X K^Z_\Omega,Y\RP\right\}
    = \underset{X \leftrightarrow Y\leftrightarrow Z}{{\rm Sym}}\left\{2 \epsilon k^Z_N\cdot{}^\Omega\sigma(X,Y)\right\},
    \\ &
      \LP\LCM_\xi ( K^\xi_\Omega),\xi\RP
    = \epsilon\LCM_{K^\xi_{\Omega}}\ln \varphi
    = -\frac{1}{2}\epsilon\xi(k^\xi_N),
    \\ &
      \LP\LCM_\xi (K^\xi_{\Omega}),X\RP
    = \epsilon \left(
          k^\xi_N\LCM_X\ln \varphi
        - \LCM_{K^X_{\Omega}}\ln\varphi
        - \frac{1}{2} X(k^\xi_N)
        - {}^\Omega\sigma(X,K^\xi_{\Omega})
    \right),
    \\ &
      \LP\LCM_X (K^\xi_{\Omega}),\xi\RP
    = \epsilon\cdot{}^\Omega\sigma(X,K^\xi_{\Omega}),
    \\ &
      \underset{X \leftrightarrow Y}{{\rm Sym}}\{
          \LP\LCM_X (K^\xi_{\Omega}),Y\RP
        + \frac{1}{2} \LP\LCM_\xi K^X_{\Omega},Y\RP
        + \epsilon \left(
              {}^\Omega\sigma(X,K^Y_{\Omega})
            - k^X_N\LCM_Y\ln \varphi
            - k^\xi_N\cdot{}^\Omega\sigma(X,Y)
        \right)
    \} = 0
\end{align}
\end{subequations}
and
\begin{align}
    \LP K^X_\Omega,Y\RP = \LP X, K^Y_\Omega \RP,\quad
    k^X_N = \epsilon\LP X,K^\xi_{\Omega}\RP,\quad
    \LCM_\xi K^X_{\Omega} = \LCM_\xi (K^X_{\Omega}) - K_{\Omega}((\LCM_\xi X)^{||}_\Omega).
\end{align}
\end{proposition}

\begin{proof}
See Appendix \ref{sec:Proposition}.
\end{proof}
 
Similar to Killing vectors, in the case of totally geodesic slices, one can lift the Killing tensor from the slice to the whole manifold and obtain a nontrivial normal component $k_N^{(\;\cdot\;)}$. This particular case of totally geodesic slices was considered, for example, in the reference \cite{Garfinkle:2010er}. Moreover, if we consider the conformal Killing tensors, a similar technique can be applied in the warped spacetimes \cite{Krtous:2015ona}, where the foliation slices are totally umbilical \cite{Chen}. In this paper, we consider the Killing tensor lift technique for arbitrary slices (not totally geodesic submanifolds). In this case, the second fundamental form ${}^\Omega \sigma$ is not trivial, and the equation (\ref{eq:killing_map_xy}) implies $k^X_N=0$, $K^X_{\Omega}=0$. Then, the family of Killing mappings $K_{\Omega}:TS_{\Omega}\rightarrow TS_{\Omega}$ can be lifted from the slices to the Killing mapping $K(\; \cdot\;)=K^{(\; \cdot\;)}_\Omega+k^{(\; \cdot\;)}_N\xi$ in the manifold $M$ with nontrivial normal components, if the following equations hold
\begin{subequations}
\begin{align}
&k^X_N=0, \quad K^\xi_{\Omega}=0,\quad \xi(k^\xi_N)=0,\quad  X(k^\xi_N)=2 k^\xi_N\LCM_X\ln \varphi-2\LCM_{K^X_{\Omega}}\ln\varphi,\label{eq:kkv1a}\\&\underset{X \leftrightarrow Y}{{\rm Sym}}\{\,\frac{1}{2}
\cdot\LP\LCM_\xi K^X_{\Omega},Y\RP+ \epsilon\cdot {}^\Omega\sigma(X,K^Y_{\Omega})-\epsilon k^\xi_N\cdot{}^\Omega\sigma(X,Y)\}=0.\label{eq:kkv1b}
\end{align}\label{eq:kkv1}
\end{subequations}

\section{Generation of a non-trivial Killing tensor}
\label{sec:generation}

Suppose that the manifold $M$ has a collection  $n\leq m-2$ of linearly independent Killing vector fields $\mK_\alpha$ tangent to the slices $S_\Omega$ of the foliation $\mF_\Omega$, i.e.
\begin{align}
\LP\mK_\alpha,\xi\RP=0,\quad \LP[\mK_\alpha,\mK_\beta],\xi\RP=0.
\end{align}
Then, such vectors $\mK_\alpha$ are also Killing vectors in the slices $S_\Omega$, and a trivial Killing mapping of the form (\ref{eq:trivial_tensor}) is always defined. Substituting this mapping into equations (\ref{eq:kkv1a}), (\ref{eq:kkv1b}) and using equations (\ref{eq:kkv2}), we obtain
\begin{subequations}
\begin{align}
    &
    X(k^\xi_N)=2 (k^\xi_N-\alpha)\LCM_X\ln \varphi, \quad
    \xi(k^\xi_N)=0, \quad
    X(\alpha)=0, \quad
    X(\gamma^{\alpha\beta})=0,\label{eq:kkv3a}
    \\&
    \label{eq:map_evolution_2}
    2\epsilon(k^\xi_N-\alpha)\cdot{}^\Omega\sigma(X,Y)=\xi(\alpha)\LP X,Y\RP
+\sum^n_{\alpha,\beta=1}\xi(\gamma^{\alpha\beta})\LP X,\mathcal K_\alpha\RP\LP\mathcal K_\beta,Y\RP, 
\end{align}
\label{eq:kkv3b}
\end{subequations}
for any $X,Y\in TS_\Omega$. There is always a trivial solution for these equations
\begin{align}
k^\xi_N=\alpha, \quad \xi(\alpha)=0, \quad \xi(\gamma^{\alpha\beta})=0,
\end{align}
corresponding to the trivial Killing tensor in $M$.  However, in some cases it can also have nontrivial solutions, which corresponds to the nontrivial Killing tensor and new conservation laws. Next, we will try to determine the necessary and sufficient conditions for the integrability of these equations and describe the methods for finding the explicit form of their solutions. 

Let us additionally assume that the Gramian matrix $\mG_{\alpha\beta}=\LP\mK_\alpha,\mK_\beta\RP$ is not degenerate. Then, we can introduce a basis $\{\mK_\alpha,e_a\}$ in $S_\Omega$ in such a way that $e_a\in\{\mK_\alpha\}^{\perp}$ with $a=1,\ldots,m-n-1$. In this basis, equation (\ref{eq:map_evolution_2}) can be cast as the following matrix system
\begin{align}
{}^\Omega\Sigma_{ab}=0,\quad {}^\Omega\Sigma_{\alpha a}=0,  \quad {}^\Omega\Sigma_{\alpha\beta}=\sum^n_{\sigma,\rho=1}\mG_{\alpha\sigma}\mG_{\beta\rho}\cdot\xi(\gamma^{\sigma\rho}),
\label{eq:kkv4}
\end{align}
where we have introduced the umbilic operator
\begin{align}
{}^\Omega\Sigma(X,Y)=2\epsilon(k^\xi_N-\alpha)\cdot{}^\Omega\sigma(X,Y)-\xi(\alpha)\LP X,Y\RP.
\end{align}
The equation ${}^\Omega\Sigma_{ab}=0$ implies that all foliation slices must be partially umbilic 
\begin{align}\label{eq:umbilic_condition}
{}^\Omega\sigma(X,Y)=h^{\Omega}\LP X,Y\RP, \quad X,Y\in\{e_a\},
\end{align}
and it is solved by
\begin{align}
k^\xi_N=\alpha+\epsilon\frac{\xi(\alpha)}{2 h^{\Omega}}.
\label{K3s}
\end{align}
In the case of strict equality $n = m-2$ the equation (\ref{eq:umbilic_condition}) is not a condition for a partial umbilic hypersurface, but only imposes a simple algebraic equation on $h^\Omega$.
Further, we will use the following relation for any arbitrary function $\omega$, constants in each slice $X\omega=0$ 
\begin{align} \label{eq:Xxi}
X\xi(\omega)=[X,\xi](\omega)=\epsilon \LP\LCM_X \xi-\LCM_\xi X,\xi\RP \xi(\omega)=\epsilon \LP X,\LCM_\xi\xi\RP\xi(\omega)=-\LCM_X \ln \varphi \cdot\xi(\omega).
\end{align}
Applying this formula for $\alpha$ and substituting (\ref{K3s}) into the first equation from (\ref{eq:kkv3a}), we find the first integrability condition
\begin{align}
X(h^{\Omega}\cdot\varphi^3) = 0.  
\label{K1}
\end{align}
Plugging (\ref{K3s}) into the second equation from (\ref{eq:kkv3a}), we can completely integrate the function $\alpha$
\begin{align}
\xi\ln\xi(\alpha)=\xi \ln h^{\Omega}-2\epsilon h^{\Omega},
\label{K2}
\end{align}
keeping in mind one more condition from (\ref{eq:kkv3a})
\begin{equation}\label{K2_cond}
    \quad X(\alpha)=0.
\end{equation}
On the other hand, using the nondegeneracy and invertibility of the Gramian matrix, one can express $\xi(\gamma^{\alpha\beta})$ from the last equation of (\ref{eq:kkv4}) as follows:
\begin{align}
\xi(\gamma^{\alpha\beta})=\sum^n_{\sigma,\rho=1}\mG^{\alpha\sigma}\mG^{\beta\rho}\cdot{}^\Omega\Sigma_{\sigma\rho}.
\label{eq:kkv5}
\end{align}

\begin{proposition} \label{prop:second_form}
Let all Killing vector fields  $\mathcal K_\alpha$ touch all slices of the foliation $S_\Omega$. Then for $\xi\in TS^{\perp}_\Omega$ 
\begin{align}
{}^\Omega\sigma(\mathcal K_\alpha,\mathcal K_\beta)=-\frac{\epsilon}{2}\cdot \xi\mG_{\alpha\beta}.
\end{align} 
\end{proposition}

\begin{proof}
Using the equation (\ref{eq:killing_equation}), we find
\begin{align}
      \LP\LCM_{\mathcal K_\alpha}\mathcal K_\beta,\xi\RP
    + \LP\LCM_{\mathcal K_\beta}\mathcal K_\alpha,\xi\RP 
    =
    -\LP\LCM_{\xi}\mathcal K_\alpha,\mathcal K_\beta\RP
    -\LP\LCM_{\xi}\mathcal K_\beta,\mathcal K_\alpha\RP
    =
    -\LCM_{\xi}\LP\mathcal K_\alpha,\mathcal K_\beta\RP.
\end{align} 
Since $\mathcal K_\alpha$ are tangent vector fields, they satisfy involutivity condition $\LP[\mathcal K_\alpha,\mathcal K_\beta],\xi\RP = 0$, and we get
\begin{align}
{}^\Omega\sigma(\mathcal K_\alpha,\mathcal K_\beta)\equiv\epsilon\LP\LCM_{\mathcal K_\alpha}\mathcal K_\beta,\xi\RP=-\frac{\epsilon}{2}\cdot\LCM_{\xi}\LP\mathcal K_\alpha,\mathcal K_\beta\RP.
\end{align} 
\end{proof}

The components ${}^\Omega\sigma(\mathcal K_\alpha,e_a)$ are equal to zero due to the equation (\ref{eq:map_evolution_2}). As a result, by combining the condition (\ref{eq:umbilic_condition}) and the Proposition \ref{prop:second_form}, the second fundamental form ${}^{\Omega}\sigma$ will be block-diagonal in the basis $\{\mK_\alpha,e_a\}$
\begin{align}\label{eq:second_form_explicit}
{}^\Omega\sigma=\begin{pmatrix}
-\frac{1}{2} \epsilon \cdot \xi\mG_{\alpha\beta} & 0\\
0 & h^\Omega \cdot  \LP e_a,e_b\RP  & 
\end{pmatrix}
\end{align}
with a mean curvature
\begin{align}
    h \equiv
    \frac{1}{\text{dim} S_\Omega} \text{Tr} {}^\Omega\sigma
    =
    - \frac{\epsilon}{2(m-1)} \cdot \xi\ln\mG
    + \frac{m-n-1}{m-1} h^\Omega.
\end{align}
Note that exactly the same structure of the second fundamental form is characteristic of partially umbilical surfaces in \cite{Kobialko:2020vqf}, which represent a generalization of photon spheres for rotating spacetime \cite{Claudel:2000yi,Teo,Galtsov:2019fzq,Galtsov:2019bty}.     

After we have established the second fundamental form in (\ref{eq:second_form_explicit}) and the $k^\xi_N$  component  in (\ref{K3s}), we can write  the umbilic operator 
\begin{align}
    {}^\Omega\Sigma_{\alpha\beta}
    =
    - \xi(\alpha) \cdot \left(
          \mG_{\alpha\beta}
        + \frac{\epsilon}{2 h^{\Omega}} \cdot \xi\mG_{\alpha\beta}
    \right).
 \label{kkv12} 
\end{align} 
Plugging this and taking into account that  $\xi k^\xi_N=0$, we can recast the equation (\ref{eq:kkv5}) into the form
\begin{align} \label{eq:xi_gamma}
    \xi\left(
          \epsilon\frac{\xi(\alpha)}{2 h^{\Omega}} \cdot \mathcal G^{\alpha\beta}
        - \gamma^{\alpha\beta}
    \right)
    = 0,
\end{align}
where we used the identity $\mG^{\alpha\lambda}\cdot\xi\mG_{\lambda\beta}=-\mG_{\lambda\beta} \cdot\xi\mG^{\alpha\lambda}$. This gives us a solution for the matrix $\gamma^{\alpha\beta}$ in terms of the new matrix $\nu^{\alpha\beta}$, which is constant along the normal:
\begin{align} \label{K3n}
    \gamma^{\alpha\beta} = 
    \epsilon\frac{\xi(\alpha)}{2 h^{\Omega}} \cdot \mathcal G^{\alpha\beta}
    - \nu^{\alpha\beta},
\end{align}
with a condition
\begin{align} \label{K3n_cond}
    \xi(\nu^{\alpha\beta}) = 0.
\end{align}

Let us find the integrability conditions for $\gamma^{\alpha\beta}$ and $\alpha$. A necessary and sufficient integrability condition on $\gamma^{\alpha\beta}$ follows from the Frobenius theorem. Acting on the equation (\ref{eq:xi_gamma}) with $X$ and  the last equation in (\ref{eq:kkv3a}), we find the commutator
\begin{equation}
    [X,\xi] \gamma^{\alpha\beta} = X\xi\left(
          \epsilon\frac{\xi(\alpha)}{2 h^{\Omega}} \cdot \mathcal G^{\alpha\beta}
    \right).
\end{equation}
Using the equations (\ref{K2}) and (\ref{eq:Xxi}) applied to $\gamma^{\alpha\beta}$ and $\alpha$, we obtain the integrability condition\footnote{Note that the order of the derivatives in this condition can be lowered if the Ricci tensor is known. See Appendix \ref{sec:Matrix}.}
\begin{align} \label{eq:gamma_integrability}
  X \left(
    \mathcal G^{\alpha\beta}
    -
    \frac{\epsilon}{2 h^{\Omega}} \cdot \xi \mathcal G^{\alpha\beta}
    \right) = 0,
\end{align}
and its trace (with respect to $\mG_{\alpha\beta}$)
\begin{align}
X\left(\frac{\epsilon}{2 h^{\Omega}}\cdot \xi\ln\mG-\ln\mG\right)+\frac{\epsilon}{2 h^{\Omega}}\cdot \xi\mG_{\alpha\beta}\cdot X\mG^{\alpha\beta}=0.
\end{align}
Finally, a non-trivial Killing tensor can be generated using the technique from the following theorem:

\begin{theorem} 
\label{KBG}
Let the manifold $M$ contains a collection of $n\leq m-2$ Killing vector fields $\mK_\alpha$ with a non-degenerate Gramian $\mG_{\alpha\beta}=\LP\mK_\alpha,\mK_\beta\RP$, tangent to the partially umbilic slices $S_\Omega$ of the foliation $\mF_\Omega$. Then, there is a nontrivial Killing tensor on manifold $M$, if the following steps can be successfully completed:

{\bf Step one}: Check  compatibility and integrability conditions (\ref{K1}), (\ref{eq:gamma_integrability}).  

{\bf Step two}: Obtain $\alpha$ from (\ref{K2}) and check the condition (\ref{K2_cond}). 

{\bf Step three}: Define $\gamma^{\alpha\beta}$ from (\ref{K3n}) using the  conditions $\xi\nu^{\alpha\beta}=0$, $X\gamma^{\alpha\beta}=0$.

{\bf Step four}: Using the functions found in the previous steps and the equations (\ref{K3s}), construct a Killing map and the corresponding Killing tensor:
\begin{align}
K(\mX)=\alpha \left(\mX-\epsilon \LP\mX,\xi \RP\xi\right)+\sum^n_{\alpha,\beta=1}\gamma^{\alpha\beta}\LP \mX,\mathcal K_\alpha\RP\mathcal K_\beta+\epsilon k^\xi_N\LP\mX,\xi \RP\xi,
\end{align}
where it is taken into account that $\mathcal{X}$ needs to be projected onto the slice in the first two terms.
\end{theorem}

\section{Connection with photon submanifolds}
\label{sec:photon}

Consider the case of a manifold with two Killing vectors spanning a timelike surface ($\epsilon=+1$). Let us define a spacelike vector field $\hat{\rho} \in \{\mathcal{K}_\alpha\}$ with index $\alpha=0,1$ numbering components in the basis\footnote{Denote a vector in $\{\mathcal{K}_\alpha\}$ in indexless hat notation, for example $\hat{\rho} = \rho^\alpha \mathcal{K}_\alpha$.} $\{\mathcal{K}_\alpha\}$ with metric $\mathcal{G}_{\alpha\beta}$ 
\begin{equation}\label{FPS1}
 \rho^\alpha=(\rho,1), \quad  \mathcal{G}_{\alpha\beta}\rho^\alpha\rho^\beta > 0,
\end{equation}
which is supposed to have constant components. We will call this vector the generalized impact parameter. In addition, we will introduce a vector $\hat{\tau}$ in $\{\mathcal{K}_\alpha\}$ orthogonal to $\hat{\rho}$:
\begin{equation}
    \tau^\alpha =\mathcal{G}^{\alpha\lambda}\epsilon_{\lambda\beta}\rho^\beta, \qquad
    \LP\hat{\tau},\hat{\tau}\RP = - \LP\hat{\rho},\hat{\rho}\RP,\qquad
    \LP\hat{\tau},\hat{\rho}\RP = 0,
\end{equation}
where $\epsilon_{\lambda\beta}$ is the two-dimensional Levi-Civita tensor. By the definition of \cite{Kobialko:2020vqf}, a timelike hypersurface tangent to two Killing fields ($\mathcal K_\alpha\in TS_\Omega$) is {\it the fundamental photon surface}, if it is umbilical for all vectors from the orthogonal complement of the field $\hat{\rho}$ (we also require the compactness of its spatial section)
\begin{align} 
    \label{FPS2a}
    {}^\Omega\sigma(\hat{\tau},\hat{\tau}) =
    h^{\Omega}\LP \hat{\tau},\hat{\tau}\RP,
    \quad
    {}^\Omega\sigma(X,Y)=h^{\Omega}\LP X,Y\RP,
    \quad
    {}^\Omega\sigma(\hat{\tau},X) = 0,
\end{align}
for $\forall X,Y\in\{e_a\}$, where $\{e_a\}$ is a basis of the orthogonal complement to $\{\mathcal{K}_\alpha\}$ in the slice introduced in Sec. \ref{sec:generation}.

\begin{proposition} 
The fundamental photon hypersurface (see Ref. \cite{Kobialko:2020vqf}) is a partially umbilical surface with a second fundamental form of the form (\ref{eq:second_form_explicit}), with the following connection between $h^\Omega$ and $\mG_{\alpha\beta}$
\begin{align}
    \rho^\alpha \mathcal{M}_{\alpha\beta} \rho^\beta = 0,
    \qquad
    \mathcal{M}_{\alpha\beta} \equiv
      \frac{1}{2 h^{\Omega}}\cdot\xi\mG_{\alpha\beta}
    - \mG_{\alpha\beta}
    - \frac{1}{2 h^{\Omega}}\cdot\xi\ln \mathcal G\cdot\mG_{\alpha\beta}.
\label{FPS3}
\end{align}
\end{proposition}

\begin{proof}
Using the Proposition \ref{prop:second_form}, the second fundamental form ${}^\Omega\sigma(\hat{\tau},\hat{\tau})$ reads
\begin{align}
    {}^\Omega\sigma(\hat{\tau},\hat{\tau}) & =
    -\frac{1}{2}\tau^\alpha \tau^\beta \xi\mG_{\alpha\beta} =
    - \rho^\alpha \rho^\beta \left(
      \frac{1}{2} \xi \mG_{\alpha\beta}
    + \mG^{\lambda\gamma} \epsilon_{\lambda\alpha}\xi\epsilon_{\gamma\beta}
    \right)=
    \frac{1}{2} \rho^\alpha \rho^\beta \left(
    - \xi \mG_{\alpha\beta}
    + \mG_{\alpha\beta} \xi \ln \mathcal{G}
    \right).
\end{align}
Substituting this expression into the first equation in (\ref{FPS2a}), we find (\ref{FPS3}).
\end{proof}
 
If the surface under consideration is totally umbilical $\mathcal{M}_{\alpha\beta}=0$, it is obviously a fundamental photon surfaces for any $\rho$. Since totally umbilical surfaces usually exist only in static spaces, and they have been considered in detail in a number of works \cite{Claudel:2000yi,Koga:2020akc}, we will focus on the case $\mathcal{M}_{\alpha\beta}\neq0$. 

Consider the foliation $\mF_\Omega$ generating a nontrivial Killing tensor in accordance with Theorem \ref{KBG}, and ask the question  whether its slice is a fundamental photon surface. First of all, we need to solve the quadratic equation (\ref{FPS3}) for $\rho$ and check the condition (\ref{FPS1}). It has nontrivial solution if the eigenvalues of the matrix $\mathcal{M}_{\alpha\beta}$ have different signs, that is $\det(\mathcal M)<0$. Then the solution for $\rho$ reads as
\begin{align}\label{PR1a}
    \rho = \frac{-\mathcal M_{01}\pm \sqrt{-\det(\mathcal M)}}{\mathcal M_{00}}.
\end{align}
Condition (\ref{FPS1}) is satisfied if the following inequality holds
\begin{align}\label{PR1b}
2(-\mathcal M_{0 1}\pm\sqrt{-\det(\mathcal M)})(\mathcal G_{01}-\mathcal G_{00}\mathcal M_{01}/\mathcal M_{00})/\mathcal M_{00}+(\mathcal G_{11}-\mathcal G_{00}\mathcal M_{11}/\mathcal M_{00})>0,
\end{align}
Equation (\ref{PR1b}) defines the so-called photon region \cite{Grenzebach,Grenzebach:2015oea}, which arises as a flow of fundamental photon surfaces \cite{Kobialko:2020vqf}. However, it has not been proven that the expression (\ref{PR1a}) for $\rho$ is constant in every slice. To understand better, let's act on Eq. (\ref{FPS3}) with $X$ 
\begin{align}
X\left\{\mathcal{M}_{\alpha\beta}\right\}\rho^\alpha\rho^\beta+2\mathcal{M}_{\alpha\beta}X(\rho^\alpha)\rho^\beta=0. 
\label{FPS4}
\end{align}
Consider the first term. We rewrite the second integrability condition with lowered indices  
\begin{align} 
  X \left( \frac{1}{2 h^{\Omega}} \cdot \xi \mathcal G_{\alpha\beta}\right) =X\mathcal G_{\alpha\beta}+ \frac{1}{2h^{\Omega}}\cdot\xi\mathcal G_{(\alpha\rho}\cdot X\mathcal G_{\lambda\beta)}\cdot\mathcal G^{\rho\lambda},
\end{align}
and plug it into the first term in (\ref{FPS4}) to get the following expression
\begin{align} \label{eq:XM1}
    &
    X \left\{\mathcal{M}_{\alpha\beta}\right\}\rho^\alpha\rho^\beta
    = \\\nonumber & =
    \left\{
          \frac{1}{h^{\Omega}}\cdot\xi\mathcal G_{\alpha\rho}\cdot X\mathcal G_{\lambda\beta}\cdot\mathcal G^{\rho\lambda}
        - X\ln\mathcal G\cdot\mG_{\alpha\beta}
        + \frac{1}{2h^{\Omega}}\cdot\xi\mathcal G_{\rho\lambda}\cdot X\mathcal G^{\rho\lambda}\cdot\mG_{\alpha\beta}
        - \frac{1}{2 h^{\Omega}}\cdot\xi\ln\mathcal G\cdot X\mG_{\alpha\beta}
    \right\}\rho^\alpha\rho^\beta,
\end{align}
where we have used $\xi \ln \mathcal G=\mathcal G^{\alpha\beta}\cdot\xi\mathcal G_{\alpha\beta}$. After subtracting the equation (\ref{FPS3}) multiplied by $X\ln\mathcal G$ from (\ref{eq:XM1}), the quantity $h{}^\Omega$ can be factored out
\begin{subequations}
\begin{align}
    X\left\{\mathcal{M}_{\alpha\beta}\right\}\rho^\alpha\rho^\beta
    \overset{\text{Eq. (\ref{FPS3})}}{=}
    \frac{1}{2 h^{\Omega}} \cdot \rho^\beta \mathcal{N}_{\beta\alpha} \rho^\alpha,
\end{align}
\begin{align}
    \mathcal{N}_{\beta\alpha} \equiv &
          \xi\mathcal G_{\rho\lambda}\cdot X\mathcal G^{\rho\lambda}\cdot\mG_{\alpha\beta}
        - \xi\mathcal G_{\rho\alpha}\cdot X\mathcal G^{\rho\lambda}\cdot \mathcal G_{\beta\lambda}
        - \xi\mathcal G_{\rho\beta}\cdot X\mathcal G^{\rho\lambda}\cdot \mathcal G_{\alpha\lambda}
        + \\\nonumber & +
          X\ln\mathcal G\cdot\xi\ln \mathcal G\cdot\mG_{\alpha\beta} 
        - \xi\ln\mathcal G\cdot X\mG_{\alpha\beta}
        - X\ln\mathcal G\cdot\xi\mG_{\alpha\beta}.
\end{align}
\end{subequations}
Raising the second index, the matrix ${\mathcal{N}_\beta}^\alpha$ can be represented as follows
\begin{subequations}
\begin{equation}
    \mathcal{N} = \mathcal{B} - \frac{1}{2} \tr(\mathcal{B}) \mathds{1},\qquad
    \mathcal{B} =
      \left(\tr(\mathfrak{a}) \mathds{1} - \mathfrak{a} \right) \mathfrak{b}
    + \left(\tr(\mathfrak{b}) \mathds{1} - \mathfrak{b} \right) \mathfrak{a},\qquad
\end{equation}
\begin{equation}
    \mathfrak{b} = \mathcal{G}_{\beta\gamma} X \mathcal{G}^{\gamma\alpha},\qquad
    \mathfrak{a} = \xi \mathcal{G}_{\beta\gamma} \cdot \mathcal{G}^{\gamma\alpha}.
\end{equation}
\end{subequations}
Since $\mathfrak{a}$, $\mathfrak{b}$ are $2\times2$ symmetric Hermitian matrices, they can be decomposed into a Pauli basis. Then the matrix $\mathcal{B}$ reads
\begin{subequations}
\begin{equation}
    \mathcal{B} =
    \frac{1}{2}\tr(\mathfrak{a}) \tr(\mathfrak{b}) \mathds{1}
    - a_i b_j \{\sigma^i \sigma^j\}
    = 
    \left(
        \frac{1}{2}\tr(\mathfrak{a}) \tr(\mathfrak{b})
        - 2 a_i b_j \delta^{ij}
    \right) \mathds{1},
\end{equation}
\begin{equation}
    \mathfrak{a}=\frac{1}{2}\tr(\mathfrak{a})\mathds{1} + a_i \sigma^i,\qquad \mathfrak{b}=\frac{1}{2}\tr(\mathfrak{b})\mathds{1} + b_i \sigma^i.
\end{equation}
\end{subequations}
We find that the matrix $\mathcal{B}$ is equal to the identity matrix, so the matrix $\mathcal{N}$ is identically zero. As a result, the first term in (\ref{FPS4}) is zero, leaving us with the second term
\begin{equation}
\mathcal{M}_{\alpha\beta}X(\rho^\alpha)\rho^\beta=0, \quad X(\rho^\alpha)=(X(\rho),0).
\end{equation}
Substituting the general solution (\ref{PR1a}) for $\rho$ into this equation, we obtain the condition $X(\rho) \sqrt{-\det(\mathcal M)}=0$. As we restrict our choice with $\det(\mathcal M) < 0$, the expression $X(\rho)$ must be zero, so $\rho$ is effectively constant for each slice. Finally, we have the following theorem.

\begin{theorem} 
Let $\mF_\Omega$ be a foliation of the manifold $M$ with non totally umbilical slices $S_\Omega$ with compact spatial section satisfying all conditions of the theorem \ref{KBG} for $\text{dim} \{\mathcal{K}_\alpha\}=2$. Then any subdomain $U_{PS} \subseteq S_\Omega$ such that the inequality (\ref{PR1b}) holds for all $p\in U_{PS}$ is a fundamental photon surface\footnote{In the case of not compact spatial section, the slice is not a fundamental photon surface by definition \cite{Kobialko:2020vqf}. However, the theorem can be generalized for such not compact surfaces.}.
\end{theorem}
 
In particular, the region $U_{PR}\subseteq M$, such that the inequality (\ref{PR1b}) holds for any point $p\in U_{PR}$, is a photon region. This theorem generalizes the connection between the existence of Killing tensors of this type and photon surfaces or spherical null geodesics, which was noted in Refs. \cite{Koga:2020akc,Pappas:2018opz,Glampedakis:2018blj}. Unfortunately, in the opposite direction, the theorem is not true, since the existence of fundamental photon surfaces does not guarantee the existence of the Killing tensor. As a counterexample, one can suggest Zipoy-Voorhees metric \cite{Kodama:2003ch} where the fundamental photon surfaces exists \cite{Galtsov:2019fzq} but there is no nontrivial Killing tensor \cite{Lukes}. Nevertheless, the existence of fundamental photon surfaces can serve as a sign that the Killing tensor can be presented in the corresponding metric, and it is advisable to check the conditions of consistency and integrability.

\section{Axially symmetric manifolds} \label{sec:axially} 
 
Consider a Lorentzian manifold $M$ with the metric tensor
\begin{align}
ds^2=-f (dt-\omega d\phi)^2+\lambda dr^2+ \beta d\theta^2 +\gamma d\phi^2. 
\end{align}
and the foliation $\mF_{\Omega}$ with timelike slices $r=\Omega$. For such a foliation, one can find the following quantities
\begin{align}
    \xi = \lambda^{-1/2}\partial_r, \quad
    h^{\Omega} = -\frac{1}{2}\lambda^{-1/2}\cdot\partial_r \ln \beta, \quad
    \varphi=\lambda^{1/2},\quad
    \mathcal G^{\alpha\beta} = \frac{1}{\gamma}\begin{pmatrix}
        \omega^2- \gamma f^{-1}  & \omega  \\
        \omega  & 1 \\
    \end{pmatrix}.
\end{align}
The compatibility and integrability conditions (\ref{K1}), (\ref{eq:gamma_integrability}) take the form
\begin{align}
    \partial_\theta(\lambda \cdot\partial_r\ln\beta)=0,\qquad
    \partial_\theta \left(
        \mathcal G^{\alpha\beta}
        +
        \frac{1}{\partial_r \ln \beta} \partial_r \mathcal G^{\alpha\beta}
    \right) = 0
\end{align}
The Eq. (\ref{K2}) can be solved as follows
\begin{align}
    \alpha=A_\theta \cdot\beta+B_\theta,
\end{align}
where the arbitrary functions $A_\theta$, $B_\theta$ depend on $\theta$ only, obeying the condition $\partial_\theta \alpha = 0$. As the result, we have one more necessary condition for the case in this section: the function $\beta$ must be of the form 
\begin{equation}\label{eq:beta_condition}
    \beta(r,\theta) = \beta_1(\theta) \beta_2(r) + \beta_3(\theta),
\end{equation}
where $\beta_{1,2,3}$ are some functions of the corresponding variables. From the equation (\ref{K3s}), the normal component is $k^\xi_N=B_\theta$. Next, we can define the matrix $\gamma$:
\begin{align}
    \gamma^{\alpha\beta} = 
    - \beta A_\theta \cdot \mathcal G^{\alpha\beta}
    - \nu^{\alpha\beta}.
\end{align}
The integrability condition guarantees that $\gamma^{\alpha\beta}$ always satisfies the equations (\ref{eq:kkv3b}) for some $\nu^{\alpha\beta}$ depending only on $\theta$. On the other hand, we have to find a $\nu^{\alpha \beta}$ that makes the equation $\partial_r \gamma^{\alpha\beta}= 0$ true. Therefore, we can omit the $\theta$-dependent part in $\gamma^{\alpha\beta}$ to some constant matrix instead of looking for $\nu^{\alpha\beta}$. Combining everything together, we get the final Killing tensor in the holonomic basis 
\begin{align}
    K^{\mu\nu} =
    \alpha g^{\mu\nu}
    + \sum_{\alpha,\beta=t,\phi}\gamma^{\alpha\beta} \mathcal{K}^\mu_\alpha \mathcal{K}^\nu_\beta
    - \beta A_\theta \lambda^{-1} \delta_r^\mu \delta_r^\nu,
\end{align}

The compatibility and integrability conditions, as well as the condition on the function $\beta$, are invariant under the multiplicative transformations of the form
\begin{align}
    \lambda\rightarrow\lambda'=u(r)\lambda, \quad \beta\rightarrow\beta'=v(\theta)\beta.
\end{align}
If $\beta$ possesses the aforementioned form (\ref{eq:beta_condition}), one can simplify the integrability condition by the substitution $\mathcal{G}^{\alpha\beta} = \tilde{\mathcal{G}}^{\alpha\beta} \cdot \beta_1 / \beta$. Then, the integrability condition is $\partial_\theta\partial_r \tilde{\mathcal{G}}^{\alpha\beta} = 0$, which is solved by $\tilde{\mathcal{G}}^{\alpha\beta} = \tilde{\mathcal{G}}_r^{\alpha\beta}(r) + \tilde{\mathcal{G}}_\theta^{\alpha\beta}(\theta)$. This generalizes the result of Ref. \cite{Johannsen:2015pca}, where the similar condition was obtained from the separability of the Hamilton-Jacobi equation. In our case, we have also included the $\beta_2(\theta)$ term. Furthermore, the compatibility condition and the function form (\ref{eq:beta_condition}) leads to the form of $\lambda=\lambda_r(r) \beta/\beta_1$, where $\lambda_r$ is an arbitrary function of $r$.

\subsection{Example: Kerr solution}

The metric for Kerr solution in the Boyer-Lindquist coordinates reads
\begin{equation}
    ds^2 =
    - f(dt - \omega d\phi)^2
    + \frac{\Sigma}{\Delta}dr^2
    + \Sigma d\theta^2
    + \Delta f^{-1}\sin^2\theta d\phi^2,
\end{equation}
\begin{equation}
    f = \frac{\Delta - a^2\sin^2\theta}{\Sigma}, \qquad
    \Sigma = r^2 + a^2 \cos^2\theta,\qquad
    \omega=\frac{-2Mar\sin^2\theta}{\Delta-a^2\sin^2\theta},\qquad
    \Delta = r(r-2M) + a^2.
\end{equation}
In the Kerr metric, $\beta=r^2+a^2\cos^2\theta$, $\lambda=\beta/\Delta$,   satisfy the compatibility condition. One can explicitly verify that $\mathcal{G}^{\alpha\beta}$ satisfies the integrability equation. In this case $\alpha=r^2$, $A_\theta=1$ and $k_N^\xi=B_\theta=-a^2\cos^2\theta$ (here we have fixed the multiplicative integration constant, which appears due to the linearity of Killing equations). The part of $\gamma^{\alpha\beta}$ independent on $\theta$ reads
\begin{align}
\gamma^{\alpha\beta}=\Delta^{-1}\left(\begin{matrix}
 (a^2+r^2)^2 &   a(a^2+r^2) \\
a(a^2+r^2) &  a^2 \\
\end{matrix}\right).
\end{align}
Finally, we get $\alpha$, $\gamma^{\alpha\beta}$ and $k_N^\xi$, which correspond to the well-known nontrivial Killing tensor for Kerr solution
\begin{align}
\alpha=r^2, \quad k^\xi_N=-a^2\cos^2\theta, \quad  \gamma^{\alpha\beta}=\Delta^{-1}\left(\begin{matrix}
 (a^2+r^2)^2 &   a(a^2+r^2) \\
a(a^2+r^2) &  a^2 \\
\end{matrix}\right).
\end{align}

\subsection{Example: Plebanski-Demianski solution}
Consider a more complicated case of the general Plebansky-Demyansky class of stationary axially symmetric solutions of type D to the Einstein-Maxwell equations with a cosmological constant. The $ds^2$ metric is read from the conformally transformed in Boyer-Lindquist coordinates 
\begin{align}
\Omega^2ds^2&=\Sigma\left(\frac{dr^2}{\Delta_r}+\frac{d\theta^2}{\Delta_\theta}\right)+\frac{1}{\Sigma}
\left((\Sigma+a\chi)^2\Delta_\theta\sin^2\theta-\Delta_r\chi^2\right)d\phi^2 \\
\label{Sol222}
&+\frac{2}{\Sigma}\left(\Delta_r\chi-a(\Sigma+a\chi)\Delta_\theta\sin^2\theta\right)dt d\phi-\frac{1}{\Sigma}
\left(\Delta_r-a^2\Delta_\theta\sin^2\theta \right)dt^2,
\end{align}
where we have defined functions
\begin{align}
\Delta_\theta &=1-a_1\cos\theta-a_2\cos^2\theta, \qquad \Delta_r=b_0+b_1r+b_2r^2+b_3r^3+b_4r^4\,,\\
\Omega &=1-\lambda(N+a \cos\theta)r, \quad \Sigma=r^2+(N+a\cos\theta)^2\,,\quad
\chi =a \sin^2\theta-2N(\cos \theta+C)\,,
\end{align}
with the following constant coefficients in $\Delta_\theta$ and $\Delta_r$:
\begin{align}
a_1 &=2aM\lambda-4aN\left(\lambda^2(k+\beta)+\frac{\Lambda}{3}\right), \quad a_2=-a^2\left(\lambda^2(k+\beta)+\frac{\Lambda}{3}\right),\quad
b_0=k+\beta, \\  
b_1 &=-2M,\quad
b_2=\frac{k}{a^2-N^2}+4MN\lambda-(a^2+3N^2)\left(\lambda^2(k+\beta)+\frac{\Lambda}{3}\right),\\
b_3 &=-2\lambda\left(\frac{kN}{a^2-N^2}-(a^2-N^2)\left(M\lambda-N\left(\lambda^2(k+\beta)+\frac{\Lambda}{3}\right)\right)\right),\quad
b_4 =-\left(\lambda^2k+\frac{\Lambda}{3}\right),\\ 
k &=\frac{1+2MN  \lambda-3N^2\left(\lambda^2\beta+\frac{\Lambda}{3}\right)}{1+3\lambda^2N^2(a^2-N^2)}(a^2-N^2), \quad 
\lambda=\frac{\alpha}{\omega}, \quad \omega=\sqrt{a^2+N^2}\,.
\end{align}
Generally, the coordinates $t$ and $r$ range over the whole $\mathbb R$, while $\theta$ and $\phi$ are the
standard coordinates on the unit two-sphere. Seven independent parameters $M,N,a,\alpha,\beta,\Lambda,C$ can be interpreted as the physical charges in the following way: $M, N$ are the mass and the NUT parameter (magnetic mass), $a$ is the Kerr-like rotation parameter, $\beta=e^2+g^2$ comprises the electric $e$ and magnetic $g$ charges, $\alpha$ is the acceleration parameter, $\Lambda$ is the cosmological constant, and thd constant $C$ defines the location of the Misner string. 
 
{\bf The first step} is to check the compatibility and integrability conditions (\ref{K1}), (\ref{eq:gamma_integrability}). The first one holds if $\alpha \cdot a=0$, i.e. either the acceleration $\alpha$ or the rotation $a$ are zero. Indeed, as  shown in Ref. \cite{Kubiznak:2007kh}, the general PD solution with acceleration possesses a conformal Killing tensor, but not the usual one. In the case $a=0$, the second condition does not hold. So, further we will consider the solution with zero acceleration $\alpha=0$, which corresponds to the dyonic Kerr-Newman-NUT-AdS solution.

For the {\bf second step} we pick up the $r$-dependent part from $\beta = \Sigma / \Delta_\theta$ for $\alpha$
\begin{align}
\beta=\frac{\Sigma}{\Delta_\theta} \quad \Rightarrow \quad \alpha=r^2, \quad A_\theta=\Delta_\theta, \quad  k^{\xi}_N=B_\theta=-\Delta^{-1}_\theta(N+a \cos \theta)^2.
\end{align}
Similarly, for the {\bf third step}, the $r$-dependent part for $\gamma^{\alpha\beta}$ is defined as
\begin{align}
\gamma^{\alpha\beta} = 
\Delta^{-1}_r \begin{pmatrix}
     P^2  & aP  \\
      aP  & a^2 \\
\end{pmatrix}
\end{align}
where $P = \Sigma + a\chi = r^2 + a^2 - 2 a C N + N^2$. In the last {\bf fourth step}, this gives us the nontrivial Killing tensor for the Kerr-Newman-NUT-AdS metric.

\subsection{Example: Kerr-Taub-NUT-AdS multicharge gauged supergravity solution}

Here we construct the Killing tensor for the metric given by Chong, Cvetic, Lu, and Pope in \cite{Chong:2004na}. The solution was obtained by starting out with the four dimensional Kerr-Taub-NUT metric, dimensionally reducing it to three dimensions with respect to time, and then lifting back after “dualization”. The metric is given by
\begin{align}
ds^2&=-\frac{\Delta_r}{a^2 W} (a dt+u_1 u_2 d\phi)^2+\frac{\Delta_u}{a^2 W}(a dt+r_1 r_2 d\phi)^2+W\left(\frac{dr^2}{\Delta_r}+\frac{du^2}{\Delta_u}\right),
\end{align}
where 
\begin{align}
&W=r_1r_2+u_1u_2, \quad r_i=r+2 M s^2_i, \quad u_i=u+2 N s^2_i,\\
&\Delta_r=r^2+a^2-2Mr +e^2r_1 r_2(r_1 r_2+a^2),\\
&\Delta_u=-u^2+a^2-2Nu +e^2u_1 u_2(u_1 u_2-a^2).
\end{align}
Here $s_i=\sinh \delta_i$ and $\delta_1$ is the magnetic charge, $\delta_2$  is the electric charge, $N$ is the NUT parameter, $a$ is the
rotation parameter, and $e$ is the gauge parameter. The cosmological constant $\Lambda$ is given by $\Lambda=-e^2$. 

{\bf The first step} is to check the compatibility conditions (\ref{K1}) for the foliations $r=\Omega$, which reads as
\begin{align}
\partial^2_{ru} W=0,
\end{align}
and is carried out identically for a given $W$. For the {\bf second step} we pick up the $r$-dependent part from $\beta = W/\Delta_u$ for $\alpha$
\begin{align}
\beta=\frac{W}{\Delta_u} \quad \Rightarrow \quad \alpha=r_1r_2, \quad A_\theta=\Delta_u, \quad  k^{\xi}_N=B_\theta=-u_1 u_2.
\end{align}
Similarly, for the {\bf third step}, the $r$-dependent part for $\gamma^{\alpha\beta}$ is presented as
\begin{align}
\gamma^{\alpha\beta} = 
\Delta^{-1}_r \begin{pmatrix}
     r^2_1  r^2_2 & a r_1  r_2   \\
      a r_1  r_2  & a^2 \\
\end{pmatrix}.
\end{align}
Note that the symmetry between $u$ and $r$ implies that there is another suitable foliation $u=\Omega$ with 
\begin{align}
&\alpha=u_1u_1, \quad  k^\xi_N=- r_1 r_2, \quad \gamma^{\alpha\beta} =\Delta^{-1}_u \begin{pmatrix}
    - u^2_1  u^2_2 & a u_1  u_2   \\
      a u_1  u_2  & -a^2 \\
\end{pmatrix}.
\end{align}
However, both foliations generate the same Killing tensor coinciding with one obtained in Ref. \cite{Vasudevan:2005bz}.

\section{Conclusions}
We have presented a purely geometric method of generating a Killing tensor in spacetimes with foliation of codimension one. For spacetimes foliated by arbitrary slices, we have derived general {\it lift equations} (\ref{eq:killing_map_xy}) relating the Killing tensor with its projection onto them and their orthogonal complement. This result includes the case of totally geodesic slices considered in Ref. \cite{Garfinkle:2010er} as a particular case. It has been demonstrated that not totally geodesic slices allow the simplification of the lift equations in the form (\ref{eq:kkv1}) as well. Using these equations one can try to lift a trivial Killing tensor defined in such slices into a nontrivial Killing tensor in the bulk. For slices tangent to Killing vectors, the lift equations can be simplified to the form (\ref{eq:kkv3b}). In this case the foliation must satisfy some consistency (\ref{K1}) and integrability (\ref{eq:gamma_integrability}) conditions, which we have derived explicitly. Furthermore, we have resolved the lift equations presenting the solution  in terms of the integrals. 

Finding a foliation suitable for integrability and consistency conditions can be challenging. However, we have found that the existence of a foliation that satisfies the integrability conditions ensures that the slices represent fundamental photon surfaces if the corresponding inequalities hold. This generalizes the result of \cite{Koga:2020akc} to the case of arbitrary stationary spaces. In the opposite direction, the existence of fundamental photon surfaces does not guarantee the existence of the Killing tensor, but may serve as an indication that the Killing tensor may exist. Therefore, it is recommended to check the consistency and integrability conditions for fundamental photon surfaces. This makes the search for fundamental photon surfaces even more important for studying the integrability of geodesic motion. 

We apply this technique to Kerr, Kerr-Newman-NUT-AdS and Kerr-NUT-AdS multicharge gauged supergravity solutions, where the Killing tensor arises purely algebraically without any differential equations have been solved. In the latter case, the discrete symmetry of the supergravity solution allows to construct the nontrivial Killing tensor using two different foliations.

The question of the relationship between fundamental photon surface and the Killing tensor of an arbitrary nature (regardless of whether it is trivial in slices or not) remains open. As the key directions of the further study, one can suggest a more general analysis of the compatibility and integrability conditions for the original system of differential equations (\ref{eq:killing_map_xy}). The presented approach can be generalized to the case of conformal Killing vector and tensor fields, which is also one of the directions for further research. 

\begin{acknowledgments}
The work was supported by the Russian Foundation for Basic Research on the projects 19-32-90095, 20-52-18012, and the Scientific and Educational School of Moscow State University “Fundamental and Applied Space Research”.
\end{acknowledgments}

\appendix

\section{Proof of the Proposition \ref{prop:killing_map}} \label{sec:Proposition}

\begin{proof}
Let us introduce here arbitrary tangent vectors $X,Y,Z\in TS_\Omega$ and a normal vector $\xi\in (TS_\Omega)^{\perp}$. 
In this proof we will make use of the identity following from the self-adjointness of mapping $K$
\begin{align}
\LP K^X_{\Omega},Y\RP=\LP X,K^Y_{\Omega}\RP, \quad k^X_N=\epsilon\LP X,K^\xi_{\Omega}\RP.
\end{align}
Now, consider the l.h.s. of Killing equations (\ref{eq:killing_mapping_equation}) projected onto $TS_\Omega$ and $\xi$ separately
\begin{subequations}
\begin{align}
    &
    \LP\LCM_\xi K(\xi),\xi\RP
    =
    \LP\LCM_\xi( K^\xi_{\Omega}+k^\xi_N\xi),\xi\RP-\LP\LCM_\xi \xi,K^\xi_\Omega+k^\xi_N\xi\RP
    = \\ &\qquad  =\nonumber
    \LP\LCM_\xi( K^\xi_{\Omega}),\xi\RP+\LP\LCM_\xi(k^\xi_N)\xi,\xi\RP-\LP\LCM_\xi \xi,K^\xi_\Omega\RP
    =
    2\epsilon\LCM_{K^\xi_{\Omega}}\ln \varphi+\epsilon \xi(k^\xi_N),
\end{align} 
\begin{align} \label{eq:projected_killing_2}
    &
    \LP\LCM_X K(\xi),\xi\RP+\LP\LCM_\xi K(X),\xi\RP+\LP\LCM_\xi K(\xi),X\RP
    = \\ & \qquad =\nonumber
      \LP\LCM_X (K^\xi_{\Omega}+k^\xi_N\xi),\xi\RP
    - \LP\LCM_X \xi,K^\xi_{\Omega} + k^\xi_N\xi\RP
    + \LP\LCM_\xi (K^X_{\Omega} + k^X_N\xi),\xi\RP
    \\&\qquad\quad\nonumber
    - \LP\LCM_\xi X,K^\xi_{\Omega}+k^\xi_N\xi\RP
    + \LP\LCM_\xi (K^\xi_{\Omega}+k^\xi_N\xi),X\RP
    - \LP\LCM_\xi \xi,K^X_{\Omega}+k^X_N\xi\RP
    = \\ &\qquad=\nonumber
      \LP\LCM_X (K^\xi_{\Omega}),\xi\RP
    + \LP\LCM_X (k^\xi_N)\xi,\xi\RP
    - \LP\LCM_X \xi,K^\xi_{\Omega}\RP
    + \LP\LCM_\xi (K^X_{\Omega}),\xi\RP
    + \LP\LCM_\xi (k^X_N)\xi,\xi\RP
    \\&\qquad\quad\nonumber
    - \LP\LCM_\xi X,K^\xi_{\Omega}\RP
    - \LP\LCM_\xi X,k^\xi_N\xi\RP
    + \LP\LCM_\xi (K^\xi_{\Omega}),X\RP
    + k^\xi_N\LP\LCM_\xi \xi,X\RP
    - \LP\LCM_\xi \xi,K^X_{\Omega}\RP
    =\\&\qquad=\nonumber
      2\epsilon\cdot{}^\Omega\sigma(X,K^\xi_{\Omega})
    + \epsilon X(k^\xi_N)
    + 2\epsilon\LCM_{K^X_{\Omega}}\ln\varphi
    - 2\epsilon k^\xi_N\LCM_X\ln \varphi
    + 2\LP\LCM_\xi (K^\xi_{\Omega}),X\RP,
\end{align} 
\begin{align}
    &
    \underset{X \leftrightarrow Y}{{\rm Sym}}\left\{\LP\LCM_X K(Y),\xi\RP+\LP\LCM_X K(\xi),Y\RP+\LP\LCM_\xi K(X),Y\RP\right\}
    =\\&\qquad=\nonumber
    \underset{X \leftrightarrow Y}{{\rm Sym}}\Big\{
          \LP\LCM_X (K^Y_{\Omega}+k^Y_N\xi),\xi\RP
        - \LP\LCM_X Y,K^\xi_{\Omega}+k^\xi_N\xi\RP
        + \LP\LCM_X (K^\xi_{\Omega}+k^\xi_N\xi),Y\RP
    \\&\qquad\qquad\qquad\nonumber
        -\LP\LCM_X \xi,K^Y_{\Omega}+k^Y_N\xi\RP
        + \LP\LCM_\xi (K^X_{\Omega}+k^X_N\xi),Y\RP
        -\LP\LCM_\xi X,K^Y_{\Omega}+k^Y_N\xi\RP
    \Big\}
    =\\&\qquad=\nonumber
    \underset{X \leftrightarrow Y}{{\rm Sym}}\Big\{
          \LP\LCM_X (K^Y_{\Omega}),\xi\RP
        + \LP\LCM_X (k^Y_N)\xi,\xi\RP
        - \LP\LCM_X Y,k^\xi_N\xi\RP
        - \LP\LCM_\xi X,k^Y_N\xi\RP
    \\&\qquad\qquad\qquad\nonumber
        + \LP\LCM_X (K^\xi_{\Omega}),Y\RP
        + k^\xi_N\LP\LCM_X \xi,Y\RP
        + \LP\LCM_\xi (K^X_{\Omega}),Y\RP
        + k^X_N\LP\LCM_\xi \xi,Y\RP
    \\&\qquad\qquad\qquad\nonumber
        - \LP\LCM_X \xi,K^Y_{\Omega}\RP
        - \LP\LCM_X Y,K^\xi_{\Omega}\RP
        - \LP\LCM_\xi X,K^Y_{\Omega}\RP
    \Big\}
    =\\&\qquad=\nonumber
    \underset{X \leftrightarrow Y}{{\rm Sym}}\Big\{
          2\LP\LCS_X (K^\xi_{\Omega}),Y\RP
        - 2\epsilon k^X_N\LCM_Y\ln \varphi
        + 2\epsilon\cdot {}^\Omega\sigma(X,K^Y_{\Omega})
        - 2\epsilon k^\xi_N\cdot{}^\Omega\sigma(X,Y)
     \\&\qquad\qquad\qquad\nonumber
        + \LP\LCM_\xi (K^X_{\Omega}),Y\RP
        - \LP\LCM_\xi X,K^Y_{\Omega}\RP\Big\},
\end{align} 
\begin{align}
    &
    \underset{X \leftrightarrow Y\leftrightarrow Z}{{\rm Sym}}\left\{\LP\LCM_X K(Y),Z\RP\right\}
    =
    \underset{X \leftrightarrow Y\leftrightarrow Z}{{\rm Sym}}\{\LP\LCM_X (K^Y_{\Omega}+k^Y_N\xi),Z\RP-\LP\LCM_X Y,K^Z_{\Omega}+k^Z_N\xi\RP\}
    \\&\qquad=\nonumber
    \underset{X \leftrightarrow Y\leftrightarrow Z}{{\rm Sym}}\{\LP\LCM_X (K^Y_{\Omega}),Z\RP+\LP\LCM_X (k^Y_N\xi),Z\RP-\LP\LCM_X Y,K^Z_{\Omega}\RP-\LP\LCM_X Y,k^Z_N\xi\RP\}
    \\&\qquad=\nonumber
    \underset{X \leftrightarrow Y\leftrightarrow Z}{{\rm Sym}}\left\{\LP\LCS_X K^Y_\Omega,Z\RP-2 \epsilon k^Z_N\cdot{}^\Omega\sigma(X,Y)\right\},
\end{align} 
\end{subequations}
where we used $\LP\LCM_\xi (k^X_N)\xi,\xi\RP=\epsilon\LP\LCM_\xi \LP X ,K^\xi_{\Omega}\RP\xi,\xi\RP=\LP \LCM_\xi X,K^\xi_{\Omega}\RP+\LP X,\LCM_\xi (K^\xi_{\Omega})\RP$ in Eq. (\ref{eq:projected_killing_2}). By virtue of Killing equations, each final expression is equal to zero.
\end{proof}

\section{Codazzi equation}
\label{sec:Matrix}

In order to simplify and generalize the calculations, we will introduce a matrix notation as follows. First of all, we choose an arbitrary basis in the subspace spanned by Killing vectors $e_\alpha\in\{\mK_\alpha\}$ and define new matrices
\begin{align}
\mathcal P\equiv\xi(\alpha)^{-1}\cdot{}^\Omega\Sigma(e_\alpha,e_\beta), \quad \mathcal M\equiv\LP \mathcal K_\alpha,e_\beta\RP, \quad \Gamma\equiv\xi(\gamma^{\alpha\beta}).
\end{align}
In particular, equations (\ref{eq:kkv5}) will be written in the form
\begin{align}
\xi(\alpha)\cdot\mathcal P=\mathcal M^{T}\Gamma\mathcal M, \quad \Gamma=\xi(\alpha)\cdot(\mathcal M^{T})^{-1}\mathcal P  \cdot\mathcal M^{-1}.
\label{kkv13}
\end{align}
Acting as before in Sec. \ref{sec:generation} and calculating the derivative of $\xi(\alpha)\cdot\mathcal P$, we obtain integrability conditions for $\mathcal P$ in matrix form
\begin{align}
    (\LCS_Z\mathcal P) = \mathcal{Q}^T + \mathcal{Q}, \quad
    \mathcal{Q} = \mathcal P\mathcal{M}^{-1}\LCS_Z\mathcal{M},
\end{align}
where $(\LCS_Z\mathcal M)_{\alpha\beta}\equiv\LP\LCS_Z \mathcal K_\alpha,e_\beta\RP$ and $(\LCS_Z\mathcal P)_{\alpha\beta}\equiv\LCS_Z(\mathcal P(e_\alpha,e_\beta))-\mathcal P(\LCS_Ze_\alpha,e_\beta)-\mathcal P(e_\alpha,\LCS_Ze_\beta)$. Using the integrability condition, Codazzi equations \cite{Chen} for $\RIM(X,\xi)$ take the form
\begin{align}
\frac{\epsilon}{h^{\Omega}}\cdot\RIM(X,\xi) =3\cdot{\rm Tr}(\mathcal P \LCS_X \ln \mathcal M)+\LCS_X \ln  h^{\Omega} \cdot\left\{{\rm Tr}(\mathcal P)+m-2\right\}.
\end{align}
Choosing the basis $e_\alpha=\mathcal K_\alpha$ (i.e. $\mathcal M=\mathcal G$), Codazzi equations reads as
\begin{align}
\RIM(X,\xi) =\frac{3}{4}\cdot\xi\mathcal G_{\alpha \beta}\cdot X\mathcal G^{\alpha \beta} +\frac{3}{2}\cdot X \ln \varphi \cdot\xi\ln\mathcal G+3\epsilon h^{\Omega} \left\{(n-m+2)\cdot X \ln \varphi-\frac{1}{2}\cdot X\ln \mathcal G\right\}.
\end{align}

\end{document}